\documentclass[12pt,onecolumn]{IEEEtran}

\IEEEoverridecommandlockouts
\usepackage{amsmath,amssymb,amsthm,mathrsfs}
\usepackage{graphicx}
\usepackage{cite}
\usepackage{flushend}
\usepackage{url}

\allowdisplaybreaks

\newcommand*{\QEDB}{\hfill\ensuremath{\square}}

\makeatletter%
\if@twocolumn%
\newcommand{\Figwidth}{\columnwidth}%
\def\twocolbreak{\nonumber\\ &}%
\else
\newcommand{\Figwidth}{4.5in}%
\def\twocolbreak{}%
\fi%
\makeatother%

\begin{document}

\title{Ergodic Fading MIMO Dirty Paper and Broadcast Channels: Capacity Bounds and Lattice Strategies}
\author{\IEEEauthorblockN{Ahmed Hindy,~\IEEEmembership{Student member,~IEEE} and 
Aria Nosratinia,~\IEEEmembership{Fellow,~IEEE}}
\thanks{The authors are with the department of Electrical Engineering, The University of Texas at Dallas, Richardson, TX 75080 USA (email: ahmed.hindy@utdallas.edu and aria@utdallas.edu).}
\thanks{The material in this paper was presented in part at the IEEE International Symposium of Information Theory (ISIT), Spain, July 2016~\cite{paper3}, and in part at the IEEE Global Communications Conference (Globecom), Washington D.C., USA, December 2016~\cite{paper4}.}
\thanks{This work was supported in part by the grant ECCS1546969 from the National Science Foundation.}
}

\maketitle




\newtheorem{theorem}{Theorem}
\newtheorem{lemma}{Lemma}
\newtheorem{remark}{Remark}
\newtheorem{corollary}{Corollary}

\def\Pw{P_w}
\def\Px{P_x}
\def\Ps{P_s}
\def\Ph{\sigma_h^2}
\def\SNR{\rho}
\def\Ix{\boldsymbol{I}}
\def\hv{\boldsymbol{h}}
\def\xv{\boldsymbol{x}}
\def\yv{\boldsymbol{y}}
\def\wv{\boldsymbol{w}}
\def\ev{\boldsymbol{e}}
\def\tv{\boldsymbol{t}}
\def\dv{\boldsymbol{d}}
\def\zv{\boldsymbol{z}}
\def\vv{\boldsymbol{v}}
\def\sv{\boldsymbol{s}}
\def\gv{\boldsymbol{g}}
\def\latpoint{\boldsymbol{\lambda}}
\def\SNRm{\boldsymbol{\Psi}}
\def\Hm{\boldsymbol{H}}
\def\Hall{\boldsymbol{H}_d}
\def\Am{\boldsymbol{A}}
\def\Bm{\boldsymbol{B}}
\def\Um{\boldsymbol{U}}
\def\Lm{\boldsymbol{L}}
\def\Gm{\boldsymbol{G}}
\def\Qm{\boldsymbol{Q}}
\def\Vm{\boldsymbol{V}}
\def\Dm{\boldsymbol{D}}
\def\comp{\mathbb{C}}
\def\real{\mathbb{R}}
\def\lat{\Lambda}
\def\voronoi{\mathcal{V}}
\def\genvoronoi{\Omega_1}
\def\codebook{\mathcal{L}_1}
\def\Ex{\mathbb{E}}
\def\vol{\text{Vol}}
\def\gap{\mathcal{G}}
\def\ball{\mathcal{B}}
\def\cov{\Sigma}
\def\covx{\boldsymbol{K_x}}
\def\Cs{\mathcal{C}}
\def\prob{\mathbb{P}}
\def\entropy{\textit{h}}
\def\info{\textit{I}}


\begin{abstract}

A multiple-input multiple-output (MIMO) version of the dirty paper channel is studied, where the channel input and the dirt experience the same fading process and the fading channel state is known at the receiver (CSIR). This represents settings where signal and interference sources are co-located, such as in the broadcast channel.  
First, a variant of Costa's dirty paper coding (DPC) is presented, whose achievable rates are within a constant gap to capacity for all signal and dirt powers. Additionally, a lattice coding and decoding scheme is proposed, whose decision regions are independent of the channel realizations. Under Rayleigh fading, the gap to capacity of the lattice coding scheme vanishes with the number of receive antennas, even at finite Signal-to-Noise Ratio (SNR). 
Thus, although the capacity of the {\em fading} dirty paper channel remains unknown, this work shows it is not far from its dirt-free counterpart.
The insights from the dirty paper channel directly lead to transmission strategies for the two-user MIMO broadcast channel (BC), where the transmitter emits a superposition of desired and undesired (dirt) signals with respect to each receiver. 
The performance of the lattice coding scheme is analyzed under different fading dynamics for the two users, showing that high-dimensional lattices achieve rates close to capacity.
\end{abstract}



\section{Introduction}
\label{intro}
 
Costa's work on the dirty paper channel~\cite{DPC} has received much interest in the past three decades. Numerous efforts have been made to derive alternative schemes with lower complexity as well as study related channel models.  Weingarten \textit{et al.}~\cite{cap_mimo_bc} proved that Costa's dirty paper coding  achieves the capacity of the MIMO broadcast channel.  Yu \textit{et al.} extended Costa's results to colored dirt. Erez \textit{et al.}~\cite{cap_lat} generalized Costa's work to states drawn from arbitrary sequences, where a lattice coding and decoding scheme was proposed. The dirty paper channel is strongly related to the discrete memoryless channel with non-causal side information at the transmitter, whose capacity is achieved using the Gel'fand-Pinsker random binning scheme~\cite{gelfand_pinsker}.
More recently, variants of the dirty paper channel have been addressed in the literature. Vaze and Varanasi~\cite{DPC_Fading_Varanasi} studied the fading MIMO dirty paper channel with partial channel state information at the transmitter (CSIT), where a scheme was derived that is optimal in the high-SNR limit. For the same setting Bennatan and Burshtein~\cite{Fading_MIMO_BC} proposed a numeric approach that achieves capacity under certain design constraints. The results in~\cite{Fading_MIMO_BC} were also applied to the fading MIMO broadcast channel. For the single-antenna fading dirty paper channel with channel state information at the receiver (CSIR), Zhang \textit{et al.}~\cite{DPC_quasi} showed that a variant of Costa's scheme is optimal at both high and low SNR.%
\footnote{In~\cite{Lattice_MIMO_DPC_C}, Lin \textit{et al.} proposed a version of the lattice coding and decoding scheme in~\cite{DMT_lattices} for the fading MIMO dirty paper channel, where a decoding rule that depends on the channel realizations is used. The proof in~\cite{DMT_lattices}, originally derived for quasi-static MIMO channels, uses the {\em Minkowski-Hlawka Theorem} to prove the existence of a codebook with negligible error probability for a given channel state. However, the existence of a universal codebook that achieves the same error probability over all channel states is not guaranteed and hence the achievable rates in~\cite{Lattice_MIMO_DPC_C} remain under question.} 
Recently, versions of the fading dirty paper channel have been studied where the signal and dirt incur different fading processes. Bergel \textit{et al.}~\cite{dpc_noncausal} proposed a lattice coding scheme for the fading dirty paper channel with non-causal noisy channel knowledge. Also, Rini and Shamai~\cite{fast_dirt} studied the scalar dirty paper channel where only the state is subject to fast fading, and considered various assumptions on channel knowledge. 
Also relevant to this discussion are studies on the ergodic broadcast channel. Li and Goldsmith derived the capacity of the $K$-user ergodic fading BC with CSIT~\cite{BC_ergodic}. For the two-user fading BC with CSIR, Tuninetti and Shamai~\cite{Shamai_BC_C} derived an inner bound for the rate region that is based on random coding and joint decoding at one of the receivers. Tse and Yates~\cite{Tse_BC} derived inner and outer bounds for the capacity, where each is a solution of an integral. Jafarian and Vishwanath~\cite{Vishwanath_BC} computed outer bounds for the case where the fading coefficients are drawn from a discrete distribution. Jafar and Goldsmith~\cite{Goldsmith_BC_CSIR} proved that in the absence of transmitter channel knowledge adding antennas to the transmitter in a BC with isotropic fading does not increase capacity.

This paper addresses the fading MIMO dirty paper channel with CSIR, where the dirt is white, stationary and ergodic. The desired signal and the dirt undergo the same fading state, which represents the case where the sources of the desired signal and interference are co-located. 
We show that dirty paper coding is within a constant gap to ergodic capacity for all SNR and all dirt power. This improves on the result in~\cite{DPC_quasi} since the gap to capacity is computed analytically for all SNR.  Moreover, a lattice coding and decoding scheme is proposed,  where the class of nested lattice codes proposed in~\cite{Erez_Zamir} are used at the transmitter, and the decision regions used are universal for almost all realizations of a given fading distribution.  
It is shown that under Rayleigh fading, the gap to the point-to-point capacity does not depend on the power of either the input signal or the state, and moreover vanishes as the number of receive antennas increases. This result has three crucial implications. First, under certain configurations, faded dirt has an insignificant effect on the capacity of the ergodic fading MIMO channel. This behavior is similar to the non-fading counterpart~\cite{DPC}. Second, the class of nested lattice codes in~\cite{Erez_Zamir} suffices to achieve rates close to capacity. Finally, a decoding rule that does not depend on the channel realizations under ergodic fading approaches optimality under mild conditions.

One advantage of the model under study is its straightforward extension to the broadcast channel with CSIR, where each receiver decodes a signal contaminated by interference stemming from the same source, and hence the desired signal and interference undergo the same fading process. We apply the dirty paper channel results to a two-user MIMO BC with different fading dynamics, where the fading process is stationary and ergodic for one receiver and quasi-static for the other receiver. In addition, the case where both users experience ergodic fading processes that are independent of each other is also studied. Unlike conventional broadcast channel techniques, the proposed scheme does not require any receiver to know the codebooks of the interference signals. Performance is compared with a version of dirty paper coding under non-causal CSIT. For the cases under study, the lattice coding scheme achieves rates very close to dirty paper coding with Gaussian inputs.

The remainder of the paper is organized as follows. In Section~\ref{sec:signal} we establish the necessary background on lattices as well as present the system model. In Section~\ref{sec:binning} an inner bound using a variant of dirty paper coding is presented and its performance is analyzed, whereas an inner bound using lattice coding and decoding is proposed in Section~\ref{sec:ptp_lat}. Application of the results to fading broadcast channels is presented in Section~\ref{sec:BC}. We conclude in Section~\ref{sec:conc}.


\section{Preliminaries}
\label{sec:signal}

\subsection{Notation and Definitions} 
\label{sec:notation}

Throughout the paper we use the following notation. Boldface lowercase letters denote column vectors and boldface uppercase letters denote matrices. The set of real and complex numbers are denoted $\real, \comp$, respectively. $\boldsymbol{A}^T, \boldsymbol{A}^H$ denote the transpose and Hermitian transpose of matrix~$\boldsymbol{A}$, respectively. 
$a_i$ denotes element~$i$ of~$\boldsymbol{a}$. $\boldsymbol{A} \succeq \Bm$ indicates that $\boldsymbol{A}-\Bm$ is positive semi-definite. $\det(\Am)$ and $\text{tr}(\Am)$ denote the determinant and trace of the square matrix~$\Am$, respectively. $\ball_n(q)$ is an $n$-dimensional ball of radius~$q$ and the volume of shape $\mathcal{A}$ is denoted $\vol(\mathcal{A})$. $\Ix_n$ is the size-$n$ Identity matrix. $\prob,\Ex$ denote the probability and expectation operators, respectively. All logarithms are in base~$2$. We define $j=\sqrt{-1}$. $\Gamma(\cdot)$ is the gamma function. $ \kappa^+$ denotes $ \max \{ \kappa,0 \}$. Real and imaginary parts of a complex number~$x$ are shown with superscripts~$x^{(R)}$ and~$x^{(I)}$. Operators $\entropy(\cdot), \, \info(\cdot \, ; \, \cdot)$ represent differential entropy and mutual information in bits, respectively.


\subsection{Lattice Codes} 
\label{sec:lattice}

A lattice $\lat$ is a discrete subgroup of $\real^n$ which is closed under reflection and real addition. 
	The fundamental Voronoi region~$\voronoi$ of the lattice~$\lat$ is the set of points with  
 minimum Euclidean distance to the origin, defined by
\begin{equation}
\voronoi = \big \{ \boldsymbol{s}: \text{arg} \min_{\latpoint \in \lat} || \boldsymbol{s} - \latpoint ||  = \boldsymbol{0} \big \}. 
\label{fund_voronoi}
\end{equation}
 
The second moment per dimension of  $\lat$ is defined as 
\begin{equation}
\sigma_{\lat}^2 = \frac{1}{n \vol(\voronoi)} \int_{\voronoi} ||\boldsymbol{s}||^2 d \boldsymbol{s},
\label{moment}
\end{equation}
and the normalized second moment $\mathit{G}(\lat)$ of $\lat$ is 
\begin{equation}
\mathit{G}(\lat) = \frac{\sigma_{\lat}^2}{\vol^{\frac{2}{n}}(\voronoi)},
\label{normalized_moment}
\end{equation}
where $\mathit{G}(\lat) > \frac{1}{2 \pi e}$ for any lattice in $\real^n$. 
Every $\boldsymbol{s} \in \real^n$ can be uniquely written as $\boldsymbol{s}= \latpoint + \ev$ where $\latpoint \in \lat$, $\ev \in \voronoi$. A lattice quantizer is then defined by
\begin{equation}
Q_{\voronoi}(\boldsymbol{s})= \latpoint \;, \quad \text{if } \boldsymbol{s} \in \latpoint+\voronoi. 
\label{quantizer}
\end{equation}
Define the modulo-$\lat$ operation corresponding to $\voronoi$ as follows
\begin{equation}
[\boldsymbol{s}] \, \text{mod} \lat \triangleq \boldsymbol{s}- Q_{\voronoi}(\boldsymbol{s}).
\label{mod}
\end{equation}
The mod~$\lat$ operation also satisfies
\begin{equation}
\big[ \boldsymbol{s} + \boldsymbol{t} \big] \, \text{mod} \lat = 
\big[ \boldsymbol{s} + [\boldsymbol{t}] \, \text{mod} \lat \big] \, \text{mod} \lat
\hspace{5mm} \forall \boldsymbol{s},\boldsymbol{t} \in \real^n.
\label{mod_mod}
\end{equation}
The lattice $\lat$ is nested in~$\lat_1$ if~$\lat \subseteq \lat_1$. 
We employ the class of nested lattice codes proposed in \cite{Erez_Zamir}. The transmitter constructs a codebook $\codebook = \lat_1 \cap \voronoi$, whose rate
is given by
\begin{equation}
R=\frac{1}{n} \log \Big( \frac{\vol(\voronoi)}{\vol(\voronoi_1) } \Big) \,.
\label{lattice_rate}
\end{equation}

The coarse lattice~$\lat$ has a second moment~$\Px$ and is good for covering and quantization, whereas the fine lattice~$\lat_1$ is good for {\em AWGN coding}, where both are construction-$A$ lattices~\cite{Loeliger,Erez_Zamir}. The existence of such lattices has been proven in~\cite{Lattices_good}.
A lattice $\lat$ is good for covering if 
\begin{equation}
\lim_{n \to \infty} \frac{1}{n} \log \bigg( \frac{\vol(\ball_n(R_c))}{\vol(\ball_n(R_f))} \bigg) =0 \, ,
\label{covering}
\end{equation}
where $R_c$, the covering radius, is the radius of the smallest sphere spanning $\voronoi$ and ${R_f}$ is the radius of the sphere whose volume is equal to $\vol(\voronoi)$. Equivalently , for a good nested lattice code with second moment $\Px$, $\voronoi$ approaches a sphere of radius $\sqrt{n \Px}$. A lattice~$\lat$ is good for quantization if
\begin{equation}
\lim_{n \to \infty} \mathit{G}(\lat) = \frac{1}{2 \pi e} \, .
\label{quantization}
\end{equation}

A key ingredient of the lattice coding schemes proposed in \cite{Erez_Zamir} is common randomness (dither)~$\dv$ at the transmitter. $\dv$ is drawn uniformly over $\voronoi$ and is known at the receiver. Some  important properties of the dither are as follows.
\begin{lemma}\cite[Lemma~1]{Erez_Zamir}
Assume an arbitrary point~$\gv \in \voronoi$ where~$\voronoi$ is the fundamental Voronoi region of a lattice~$\lat$, and a point~$\dv \in \voronoi$  drawn according to a uniform distribution over the region~$\voronoi$. If $\gv$ is independent of $\dv$, then $\xv \triangleq [\gv - \dv] \, \text{mod} \lat$ is also uniformly distributed over $\voronoi$ and independent of~$\gv$. \QEDB
\label{lemma:uniform}
\end{lemma}

A lattice with optimal quantizer $Q_{\voronoi}^{\text{opt}}$ is one whose dither has the minimum normalized second moment~$\mathit{G}(\lat)$. The dither of such lattice would then have the following property.

\begin{lemma}\cite[Theorem 1]{Lattice_quantization}.
The dither~($\dv_{\text{opt}}$) of a lattice with optimal quantizer of second moment~$\sigma_{\lat}^2$ is white with autocorrelation $\Ex [ \dv_{\text{opt}} \dv_{\text{opt}}^T ]= \sigma_{\lat}^2  \Ix_n$.  \QEDB
\label{lemma:x_distribution}
\end{lemma}

From~\eqref{quantization}, since the proposed class of lattices is good for quantization, i.e., has minimum normalized second moment asymptotically, the autocorrelation of~$\dv$ approaches that of~$\dv_{\text{opt}}$, implying that $\Ex [ \dv \dv^T ] \to \sigma_{\lat}^2  \Ix_n$ as $n \to \infty$. 
For a more comprehensive review on lattice codes see~\cite{lattice}.


\subsection{System Model} 
\label{sec:system}

Consider a MIMO point-to-point channel with Gaussian noise and $M,N$ antennas at the transmitter and receiver sides, respectively. The fading process is stationary and ergodic, where the random channel matrix is denoted by~$\Hm$. The received signal is impeded not only by Gaussian noise, but also by an undesired signal $\sv$ that experiences the same fading as the desired signal $\xv$, as follows
\begin{equation}
\yv_i=\Hm_i \xv_i + \Hm_i \sv_i + \wv_i,
\label{sig_Rx}
\end{equation}
where the channel coefficient matrices $\Hm_i \in \comp^{N \times M}$ at time $i=1,\ldots,n$ denote realizations of the random matrix~$\Hm$. Moreover,~$\Hm$ is zero mean and isotropically distributed, i.e., $P(\Hm)=P(\Hm \Vm)$ for any unitary matrix $\Vm$ independent of $\Hm$. The unordered eigenvalues of the Hermitian random matrix~$\Hm^H \Hm$, denoted by~$\sigma_1^2, \sigma_2^2, \ldots, \sigma_M^2$, are also random, and their distribution is characterized by the distribution of~$\Hm$.
The receiver has instantaneous channel knowledge, whereas the transmitter only knows the channel distribution. $\xv_i \in \comp^M$ is the transmitted vector at time~$i$, where the codeword
\begin{equation}
\xv \triangleq [\xv_1^H \, \xv_2^H , \ldots \xv_n^H]^H
\label{codeword_MIMO}
\end{equation} 
is transmitted throughout~$n$ channel uses and satisfies $\Ex [ || \xv ||^2 ] \leq n \Px$. 
The noise~$\wv \in \comp^{N n}$ defined by $\wv \triangleq [\wv_1^H , \ldots , \wv_n^H]^H$ is a circularly-symmetric zero-mean i.i.d. Gaussian noise vector with covariance $ \Pw \Ix_{N n} $, and is independent of~$\Hm$.
$\sv \in \comp^{M n}$, where $\sv \in [\sv_1^H , \ldots, \sv_n^H ]^H$  represents the state (dirt) that is independent of $\Hm,\wv$ and is known non-causally at the transmitter. Unless otherwise stated, we assume $\sv$ is a stationary and ergodic sequence whose elements have zero-mean and variance~$\Ps$.

An intuitive outer bound for the rates of the channel in~\eqref{sig_Rx} would be the point-to-point channel capacity in the absence of the state~$\sv$, as follows
\begin{equation}
C \leq  \Ex \Big[ \log \det \big ( \Ix_M +  \frac{\Px}{M \Pw} \Hm^H \Hm \big) \Big].
\label{capacity_outer}
\end{equation}
Had the channel coefficients been known non-causally at the transmitter, the rate in~\eqref{capacity_outer} would have been achieved in a straightforward manner from Costa's result since the new state~$\tilde{\sv}_i \triangleq \Hm_i \sv_i$ would be known at the transmitter~\cite[Chapter 9.5]{network_info}. However, in the present model $\Hm$ is unknown at the transmitter, causing the problem to become more challenging.
In the sequel, two different inner bounds are studied that approach the outer bound in~\eqref{capacity_outer}.%
\footnote{The received signal in~\eqref{sig_Rx} is analogous to that received in a broadcast channel setting, where the desired and interfering signals~$\xv,\sv$ respectively, stem from the same source. This observation is key to understanding the strong connection between the dirty paper and broadcast channels, which is exploited later in Section~\ref{sec:BC}.}


\section{Dirty Paper Coding Inner Bound}
\label{sec:binning}

In this section we aim at deriving a capacity-approaching scheme for the dirty paper channel in~\eqref{sig_Rx}. This channel is a variation of Gel'fand and Pinsker's discrete memoryless channel with a state known non-causally at the transmitter~\cite{gelfand_pinsker},  whose capacity can be expressed by
\begin{equation}
C = \max_{P_{V,X|S}} \info(V;Y,H) - \info(V;S), 
\label{capacity_dmc}
\end{equation} 
where~$S$ represents the state and~$V$ is an auxiliary random variable. $Y$ and $H$ represent the receiver observation and the available CSIR, respectively. Unfortunately the single-letter capacity optimization in~\eqref{capacity_dmc} is not tractable. In~\cite{DPC}, Costa studied the non-fading dirty paper channel with single antenna and additive Gaussian noise, where he showed that  the point-to-point capacity can be achieved and the impact of~$\sv$ can be entirely eliminated. The ingredients of the achievable scheme are using Gaussian codebooks that are correlated with the known ``dirt'' in conjunction with typicality decoding. In the sequel a similar approach is adopted for the fading MIMO dirty paper channel.


\begin{theorem}
For the ergodic MIMO dirty paper channel in~\eqref{sig_Rx} with i.i.d.\ Gaussian dirt, any rate satisfying
\begin{equation}
R \leq  \Big (  \Ex \big[ \log \det {( \frac{\Px}{\Px + M \Ps} \Ix_{M} +  \frac{\Px}{M \Pw} \Hm^H \Hm )} \big]  \Big ) ^+,
\label{capacity_inner_MIMO}
\end{equation}
is achievable.
\label{theorem:ptp_MIMO}
\end{theorem}

\begin{proof} 
We first consider real-valued channels. We follow in the footsteps of the encoding and decoding schemes in~\cite{DPC,DPC_Fading_Varanasi}, where random binning at the encoder and typicality decoding were proposed. Details are as follows.

{\em Encoding:} The transmitted signal of length $Mn$ is in the form $\xv = \vv- \Um \sv$, where $\vv$ is drawn from a codebook consisting of $2^{n \tilde{R}}$ codewords for some $\tilde{R}>0$. The codewords are drawn from a Gaussian distribution with zero mean and covariance~$\frac{\Px}{M} \Ix_{Mn} + \Ps \Um \Um^T$, where  $\Um \in \real^{Mn \times Mn}$ will be determined later. These codewords are randomly assigned to $2^{nR}$ bins for some $0 < R < \tilde{R}$, so that each bin will contain approximately $2^{n(\tilde{R}-R)}$ codewords. As long as~$2^{n (\tilde{R}-R)} > 2^{n \info(V;S)}$, typicality arguments guarantee the existence of a codeword $\vv_0$ in each bin that is jointly typical with the state $\sv_0$, i.e., $\vv_o - \Um \sv_o$ is nearly orthogonal to $\sv_o$~\cite{DPC}. The bin index is chosen according to the message to be transmitted, and from that bin the appropriate codeword is transmitted that is jointly typical with the state. The transmitter emits~$\xv_o \triangleq \vv_o - \Um \sv_o$, which satisfies the power constraint, $\Ex [||\xv||^2] \leq n \Px$ (recall $\xv_o,\sv_o$ are orthogonal).

{\em Decoding:} Given the occurrence of state~$\sv_o$, the received signal is given by
\begin{align}
\yv =& \Hm_d \xv_o + \Hm_d \sv_o + \wv_o \nonumber \\
 =& \Hm_d \vv_o + \Hm_d (\Ix_{Mn}-\Um)  \sv_o + \wv,
\label{dirty_rx_2}
\end{align}
where $\Hm_d \triangleq \text{diag} \{ \Hm_1, \ldots, \Hm_n \}$, and the receiver knows the codebook of~$\vv$. From standard typicality arguments,~$\vv_o$ can be decoded reliably as long as~$2^{n \tilde{R}} < 2^{n \info(V;Y,H)}$ at large~$n$.

{\em Rate analysis:}
Based on the encoding and decoding procedures, the number of distinguishable messages that can be transmitted is equal to the number of bins~$2^{nR}$. The rate can then be analyzed as follows 
\begin{align}
n R <& \, \info(V;Y,H) - \info(V;S)   \nonumber \\
	=& \, \entropy(V) + \entropy(Y|H) - \entropy(Y,V|H) - \entropy(V) + \entropy(V|S)  \nonumber \\
	=& \, \entropy(Y|H) + \entropy(V- \alpha S|S) - \entropy(Y,V|H) \nonumber \\
	=& \, \entropy(Y|H) + \entropy(X) - \entropy(Y,V|H)  \nonumber \\
	=& \, \frac{1}{2}  \log \big( (2 \pi e)^{Mn} \det { (\Qm)} \big) 
	\twocolbreak
	  +\frac{1}{2} \log \big( (2 \pi e)^{Mn} \det {(  \frac{\Px}{M} \Ix_{Mn})} \big)     \nonumber \\
	 & \,  -\frac{1}{2} \log \big( (2 \pi e)^{(M+N)n}  \det {( \Qm)}      \det \big( \frac{\Px}{M} \Ix_{Mn} + \Ps \Um \Um^T    \nonumber \\
				& ~~~   - (\frac{\Px}{M} \Ix_{Mn} + \Ps \Um) \Hall^T  \Qm^{-1} \Hall (\frac{\Px}{M} \Ix_{Mn} + \Ps \Um^T)  \big) \big),
\end{align}
where $\Qm \triangleq  \frac{\Px}{M} \Hall \Hall^T + \Ps \Hall \Hall^T + \Pw \Ix_{Nn} $.
On choosing $\Um= \Ix_{Mn}$,
\begin{equation}
R	< \frac{1}{2n} \sum_{i=1}^n \log \det {\big( \frac{\Px}{\Px + M \Ps} \Ix_{M} + \frac{\Px}{M \Pw} \Hm_i^T  \Hm_i \big)} \Big].  
\label{inner_LLN}
\end{equation}
From the law of large numbers, \eqref{inner_LLN} converges to
\begin{equation}
R	<  \frac{1}{2} \Ex \big[ \log \det {\big( \frac{\Px}{\Px + M \Ps} \Ix_{M} + \frac{\Px}{M \Pw} \Hm^T  \Hm \big)} \big] - \epsilon
\label{rate_LLN}
\end{equation}
 with probability~1, where~$\epsilon$ vanishes as~$n \to \infty$. Finally, the result can be extended to complex-valued channels through the following equivalence
\begin{equation}
\tilde{\yv}_i=\tilde{\Hm}_i \tilde{\xv}_i + \tilde{\Hm}_i \tilde{\sv}_i + \tilde{\wv}_i,
\label{sig_Rx_complex}
\end{equation}
where
\begin{equation}
\tilde{\Hm}_i \triangleq 
\begin{bmatrix}
 \,  \Hm_i^{(R)} \hspace{4mm}  - \Hm_i^{(I)}  \\
 \, \Hm_i^{(I)} \hspace{4mm}  ~~~~  \Hm_i^{(R)}
\end{bmatrix}
\label{channel_complex}
\end{equation}
 is a $2N \times 2M$ channel matrix and $\tilde{\xv}_i \triangleq [ \xv_i^{(R) T} \xv_i^{(I) T} ]^T$ and similarly for $\tilde{\sv}_i,\tilde{\wv}_i$. The rate achieved can then be expressed by~\eqref{capacity_inner_MIMO}.  This concludes the proof of Theorem~\ref{theorem:ptp_MIMO}.
\end{proof}

In the following we bound the gap between the inner and outer bounds in~\eqref{capacity_outer},\eqref{capacity_inner_MIMO}. For ease of exposition we assume~$M \leq N$.

\begin{corollary}
The rate achieved in~\eqref{capacity_inner_MIMO} is within $M$ bits of the capacity.
\label{cor:gap_MIMO}
\end{corollary}

\begin{proof}
The gap between the capacity outer bound in~\eqref{capacity_outer} and~~\eqref{capacity_inner_MIMO} is bounded by
\begin{align}
\gap \triangleq & \, C-R  \nonumber  \\
\leq &  \Ex \big[ \log \det {( \Ix_{M} +  \frac{\Px}{M \Pw} \Hm^H \Hm )} \big] \, 
\twocolbreak
-  \Big (  \Ex \big[ \log \det {( \frac{\Px}{\Px + M \Ps} \Ix_{M} +  \frac{\Px}{M \Pw} \Hm^H \Hm )} \big] \Big )^+ \nonumber \\
\leq &  \Ex \big[ \log \det {( \Ix_{M} +  \frac{\Px}{M \Pw} \Hm^H \Hm )} \big] \,   
\twocolbreak
   - \, \Ex \big[ \log \det {( \frac{\Px}{\Px + M \Ps} \Ix_{M} +  \frac{\Px}{M \Pw} \Hm^H \Hm )} \big]  \nonumber \\
= & \,   \Ex \big[ \sum_{j=1}^{M}  \log ( 1 +  \frac{\Px}{M \Pw} \sigma_j^2 ) \big] \, 
\twocolbreak
-  \,  \Ex \big[ \sum_{j=1}^{M}  \log ( \frac{\Px}{\Px + M \Ps} +  \frac{\Px}{M \Pw} \sigma_j^2 ) \big]  \label{eq_gap_1} \\
= &   \sum_{j=1}^{M} \Ex_{\sigma_j} \Big[   \log ( 1 +  \frac{\Px}{M \Pw} \sigma_j^2 )  -  \log ( \frac{\Px}{\Px + M \Ps} +  \frac{\Px}{M \Pw} \sigma_j^2 ) \Big]  \nonumber \\
= &   \sum_{j=1}^{M} \bigg( \prob \big(\frac{\Px}{M \Pw} \sigma_j^2 \geq 1 \big) \Ex \Big[   \log ( 1 +  \frac{\Px}{M \Pw} \sigma_j^2 ) 
\twocolbreak
~~ -  \log ( \frac{\Px}{\Px+ M \Ps} +  \frac{\Px}{M \Pw} \sigma_j^2 ) \Big| \frac{\Px}{M \Pw} \sigma_j^2 \geq 1 \Big]    \nonumber \\
 & +  \prob \big(\frac{\Px}{M \Pw} \sigma_j^2 < 1 \big) \Ex \Big[   \log ( 1 +  \frac{\Px}{M \Pw} \sigma_j^2 ) 
\twocolbreak
~~ -  \log ( \frac{\Px}{\Px+ M\Ps} +  \frac{\Px}{M \Pw} \sigma_j^2 ) \Big| \frac{\Px}{M \Pw} \sigma_j^2 < 1 \Big]  \bigg)  \label{eq_gap_2} \\
< &   \sum_{j=1}^{M} \bigg( \prob \big(\frac{\Px}{M \Pw} \sigma_j^2 \geq 1 \big) \Ex \Big[   1 + \log ( \frac{\Px}{M \Pw} \sigma_j^2 )  
\twocolbreak
~~ -  \log ( \frac{\Px}{\Px+ M \Ps} +  \frac{\Px}{M \Pw} \sigma_j^2 ) \Big| \frac{\Px}{M \Pw} \sigma_j^2 \geq 1 \Big]   \nonumber \\
 & +  \prob \big( \frac{\Px}{M \Pw} \sigma_j^2 < 1 \big) \Ex \Big[ \,  1 \,
\twocolbreak
~~ -  \log ( \frac{\Px}{\Px+M \Ps} +  \frac{\Px}{M \Pw} \sigma_j^2 ) \Big| \frac{\Px}{M \Pw} \sigma_j^2 < 1 \Big]  \bigg)  \nonumber \\
< &   \sum_{j=1}^{M} \Big( \prob \big( \frac{\Px}{M \Pw} \sigma_j^2 \geq 1 \big) \Ex \big[  \, 1 \,  \big| \frac{\Px}{M \Pw} \sigma_j^2 \geq 1 \big] 
\twocolbreak
+  \prob \big( \frac{\Px}{M \Pw} \sigma_j^2 < 1 \big) \Ex \big[  \, 1 \,   \big| \frac{\Px}{M \Pw} \sigma_j^2 < 1 \big]  \Big) \, = \, M \, ,
\label{eq:gap_MIMO}
\end{align}
where $\sigma_j^2$ are the eigenvalues of $\Hm^H \Hm$ for $j=1,\ldots, M$, as explained in Section~\ref{sec:system}, and hence~\eqref{eq_gap_1} is an alternative representation of  the expressions in~\eqref{capacity_outer},\eqref{capacity_inner_MIMO} in terms of the channel eigenvalues. \eqref{eq_gap_2} follows from the law of total expectation. The gap to capacity is then shown to be bounded from above by~$M$ bits.
\end{proof}

\begin{remark}
The rates achieved in~\cite[Section IV]{DPC_Fading_Varanasi} were shown to approach capacity at high SNR, i.e., $\gap \to 0$ as $\Px \to \infty$. Meanwhile, Corollary~\ref{cor:gap_MIMO} bounds the gap to capacity within a constant number of bits, irrespective of the values of $\Px,\Ps$ as well as the fading distribution. This result does not contradict with that in~\cite{DPC_Fading_Varanasi}, however. For instance, when $M=N=1$ and $\Pw=1$,  the gap to capacity would be as follows
\begin{align}
\gap \, \leq \, & \Ex \big[ \log ( 1 + \Px |h|^2 ) \big]  -  \Ex \big[ \log ( \frac{\Px}{\Px + \Ps} +  \Px |h|^2 ) \big]    \nonumber \\ 
 = \, & \Ex \big [ \log (1 + \frac{\frac{\Ps}{\Px+\Ps}}{\frac{\Px}{\Px+\Ps} + |h|^2 \Px})  \big] \,  
 < \, \Ex \big[ \log (1 + \frac{1}{|h|^2 \Px})  \big],  
\label{gap_high}
\end{align}
which vanishes as $\Px \to \infty$, confirming the result in~\cite{DPC_Fading_Varanasi}.
\end{remark}

%
%


\section{Lattice coding inner bound} 
\label{sec:ptp_lat}

Although the scheme proposed in Section~\ref{sec:binning} achieves rates that are close to capacity, it has large computational complexity at both the transmitter and receiver since it uses Gaussian codebooks. In this section a lattice coding and decoding scheme is proposed that transmits a dithered lattice codeword and at the receiver uses a single-tap equalizer and Euclidean lattice decoding. In our scheme the use of CSIR is limited to the equalizer; the decision rule does not depend on the instantaneous realizations of the fading channel.

\begin{theorem}
For the ergodic fading MIMO dirty paper channel given in~\eqref{sig_Rx}, any rate 
\begin{equation}
R < \bigg ( -  \log \det \Big( \Ex \big[ \big( \frac{\Px}{\Px+M \Ps} \Ix_{M} + \frac{\Px}{M \Pw} \Hm^H \Hm  \big)^{-1}  \big] \Big) \bigg ) ^+
\label{R_lat}
\end{equation}
is achievable using lattice coding and decoding. 
\label{theorem:ptp_lat}
\end{theorem}

\begin{proof}
We first consider real-valued channels. The encoder is designed as follows

{\em Encoding:} Nested lattice codes are used where $\lat \subseteq \lat_1$. The transmitter emits a signal $\xv$ as follows
\begin{align}
\xv &= [\tv - \Bm \sv - \dv] \text{ mod} \lat,  \nonumber \\
&= \tv + \latpoint  - \Bm \sv - \dv,
\label{sig_tx_lat}
\end{align}
where $\tv \in \codebook$ is drawn from a nested lattice code with $\lat \subseteq \lat_1$, dithered with~$\dv$ which is drawn uniformly over~$\voronoi$, and $\latpoint= - Q_{\voronoi}(\tv - \Bm \sv - \dv ) \in \lat$ from~\eqref{mod}. $\Bm \in \real^{Mn \times Mn}$ will be determined in the sequel.\footnote{Note that~$\Bm$ must be independent of the channel realizations.}
Note from Lemma~\ref{lemma:uniform} that the dither guarantees the independence of $\xv$ from both~$\tv$ and~$\sv$.

{\em Decoding:} The received signal in~\eqref{sig_Rx} undergoes single-tap equalization by~$\Um_d$ and the dither is removed as follows
\begin{align}
\yv' =& \Um_d^T \yv + \dv \nonumber \\
 =& \xv + (\Um_d^T \Hm_d - \Ix_{Mn}) \xv + \Um_d^T \Hm_d \sv + \Um_d^T \wv + \dv \nonumber \\ 
 =& \tv + \latpoint + \zv,
\label{sig_rx_2}
\end{align}
where 
\begin{equation}
\zv \triangleq (\Um_d^T \Hm_d - \Ix_{Mn}) \xv + (\Um_d^T \Hm_d - \Bm) \sv + \Um_d^T \wv
\label{eq_noise}
\end{equation}
is independent of~$\tv$ according to Lemma~\ref{lemma:uniform}. 
For the special case $\Hm_d = \Ix_{Mn}$, the problem reduces to the non-fading dirty paper channel, where the choices $\Um_d= \Bm = \frac{\Px}{\Px + \Pw} \Ix_{Mn}$ are optimal and the point-to-point channel capacity can be achieved via the lattice coding and decoding scheme in~\cite{cap_lat}. However, this  scheme cannot be directly extended to ergodic fading, since the channel realizations are unknown at the transmitter. We choose~$\Bm= \Ix_{Mn}$ so that 
\begin{equation}
\zv \triangleq (\Um_d^T \Hm_d - \Ix_{Mn}) \, (\xv + \sv) + \Um_d^T \wv.
\label{eq_noise_2}
\end{equation}
The motivation behind this choice is aligning the dirt and self-interference terms in~\eqref{eq_noise}. The equalization matrix~$\Um_d$ is designed to minimize~$\Ex \big[ ||\zv||^2 \big]$. $\Um_d$ would then be a block-diagonal matrix whose diagonal block~$i$, $\Um_i$ is given by\footnote{Unlike~\cite{Erez_Zamir,cap_lat}, $\Um_i$ is {\em not} the MMSE equalizer for the channel in~\eqref{sig_Rx}. }
\begin{equation}
\Um_i  = (\frac{\Px}{M}+\Ps)  \big((\frac{\Px}{M}+\Ps) \Hm_i \Hm_i^T + \Pw \Ix_{N} \big)^{-1} \Hm_i.
\label{eq:U_MSE}
\end{equation}

From~\eqref{eq_noise_2} and \eqref{eq:U_MSE}
\begin{equation}
\zv_i = \big ( \Ix_M + \frac{1}{\Pw} (\frac{\Px}{M}+\Ps) \Hm_i^T \Hm_i  \big)^{-1} \, (\xv_i + \sv_i) 
 + (\frac{\Px}{M}+\Ps) \Hm_i^T ( (\frac{\Px}{M}+\Ps) \Hm_i \Hm_i^T + \Pw \Ix_{N} )^{-1} \wv_i.
\label{eq:zi}
\end{equation}

Naturally, the distribution of $\zv$ conditioned on $\Hm_i$ (which is known at the receiver) varies across time.  This variation produces complications, so in order to simplify the decoding process,  we ignore the instantaneous channel knowledge at the decoder following the equalization step, i.e., after equalization the receiver considers $\Hm_i$ a random matrix.


\begin{lemma}
\label{lemma:z_dist}

For any $\epsilon>0$ and $\gamma>0$, there exists $n_{\gamma,\epsilon}$ such that for all $n>n_{\gamma,\epsilon}$, %

\begin{equation}
   \prob \big( \zv \notin \genvoronoi \big) < \gamma,
\label{eq:error_event}
\end{equation}
where $\genvoronoi$ is the following sphere
\begin{align}
\genvoronoi \triangleq \Big\{ \vv & \in \real^{M n} \, : \, ||\vv||^2 \leq  (1+\epsilon)  n \, 
 \text{tr} \big( \Ex \big [ \boldsymbol{\bar{\cov}} ] \big)  \Big\}, 
\label{voronoi_MIMO}
\end{align} 
and $\boldsymbol{\bar{\cov}} \triangleq \Ex \Big [ \big( \frac{1}{\frac{\Px}{M}+\Ps} \Ix_{M} + \frac{1}{\Pw} \Hm^T \Hm \big)^{-1} \Big ]$ is a scaled identity matrix.
\end{lemma}

\begin{proof}
See Appendix~\ref{appendix:z_dist}.
\end{proof}

We apply a version of the ambiguity decoder proposed in~\cite{Loeliger} with a spherical decision region~$\genvoronoi$ in~\eqref{voronoi_MIMO}. 
 The decoder decides $\hat{\tv} \in \lat_1$ if and only if the received point falls exclusively within the decision region of~$\hat{\tv}$, i.e., $\yv' \in \hat{\tv} + \genvoronoi$.

{\em Probability of error:} As shown in~\cite[Theorem~4]{Loeliger}, on averaging over the set of all fine lattices~$\Cs$ of rate~$R$ whose construction follows that in Section~\ref{sec:lattice}, the probability of error can be bounded by 
\begin{align}
\frac{1}{|\Cs|} \sum_{\Cs_i \in \Cs} \, \prob_e <& \, \prob (\zv \notin \genvoronoi) + (1+ \delta) \, \frac{\vol(\genvoronoi)}{\vol(\voronoi_1)} \nonumber \\
=& \, \prob (\zv \notin \genvoronoi) + (1+ \delta) 2^{nR} \, \frac{\vol(\genvoronoi)}{\vol(\voronoi)}, 
\label{error_prob}
\end{align}
for any $\delta > 0$, and the equality follows from~\eqref{lattice_rate}. This is a union bound on two events: the received vector residing outside the correct decision region, or in the intersection of two distinct decision regions, i.e.,
$\big\{\yv' \in \{ \tv_1 + \genvoronoi \} \cap \{ \tv_2 + \genvoronoi
\} \big\}$. 
From Lemma~\ref{lemma:z_dist}, the first term in~\eqref{error_prob} vanishes with~$n$.  
Define~$\SNRm \triangleq \frac{\Px}{M} \boldsymbol{\bar{\cov}}^{-1}$.
The volume of~$\genvoronoi$ is then given by
\begin{equation}
\vol(\genvoronoi) = (1+ \epsilon)^{\frac{M n}{2}} 
 \vol \big( \ball_{M n} (\sqrt{ n \Px}) \big) \,  \det \big ( \SNRm^{\frac{-1}{2}} \big) .
\label{vol_omega_MIMO}
\end{equation}

The second term in~\eqref{error_prob} is bounded by
\begin{align}
& (1+ \delta) 2^{nR} (1+ \epsilon)^{M n/2} \frac{\vol(\ball_{M n} (\sqrt{ n \Px}))}{\vol(\voronoi)} 
 \, \det \big ( \SNRm^{\frac{-1}{2}} \big ) \,  \twocolbreak
= \, (1+ \delta) 2^{ - M n \Big ( - \frac{1}{M n} \log \big( \frac{\vol(\ball_{M n} (\sqrt{ n \Px}))}{\vol(\voronoi)} \big) + \xi \Big ) }, 
\label{eq:exponential_MIMO}
\end{align}
where 
\begin{align}
\xi \, \triangleq & \,  \frac{-1}{2} \log({1+ \epsilon})  - \frac{1}{M} R -  \frac{1}{2 M n} \log \det ( \SNRm^{-1} )    \nonumber \\
= & \, \frac{-1}{2} \log({1+ \epsilon})   - \frac{1}{M} R  
\twocolbreak
- \frac{1}{2 M}  \log \det \Big( \Ex \big[ \big( \frac{\Px}{\Px+M \Ps} \Ix_{M} + \frac{\Px}{M \Pw} \Hm^T \Hm  \big)^{-1}  \big] \Big).
 \label{xi_MIMO}
\end{align}
From~\eqref{covering}, since the lattice $\lat$ is good for covering, the first term of the exponent in~\eqref{eq:exponential_MIMO} vanishes. From~\eqref{eq:exponential_MIMO}, $\prob_e  \to 0$ as $n \to \infty$ if $\xi>0$, which is satisfied when
\begin{align*}
R < \, & \frac{-1}{2}  \log \det \Big( \Ex \big[ \big( \frac{\Px}{\Px+ M \Ps} \Ix_{M} + \frac{\Px}{M \Pw} \Hm^T \Hm  \big)^{-1}  \big] \Big)
\twocolbreak
- \frac{1}{2} \log({1+ \epsilon}) - \epsilon',
\end{align*}
where~$\epsilon,\epsilon'$ vanish with~$n$.
The existence of at least one lattice~$\Cs_i \in \Cs$ that achieves~\eqref{error_prob} is straightforward. For the coarse lattice, any covering-good lattice from~$\Cs$ with second moment~$\Px$ can be picked, e.g., a pair of self-similar lattices can be used for the nested lattice.
In the event of successful decoding, from~\eqref{sig_rx_2} the outcome of the decoding process would be  $\hat{\tv} = \tv + \latpoint$. On applying the modulo-$\lat$ operation on~$\hat{\tv}$,
\begin{equation}
[\hat{\tv}] \, \text{mod}\lat \, = \, [\tv+\latpoint] \, \text{mod}\lat \, = \, \tv,
\label{mod_lambda}
\end{equation}
where the second equality follows from~\eqref{mod_mod} since~$\latpoint \in \lat$. This concludes the proof for real-valued channels.

For complex-valued channels, we follow in the footsteps of~\cite[Theorem 2]{paper1}, and hence only a sketch of the proof is provided. With a slight abuse of notation, we denote the complex-valued elements by a superscript~$^{\sim} \,$.

{\em Encoding:} Since the channel is complex-valued, {\em two} independent codewords are selected from the same nested lattice code~$\tilde{\lat}_1 \supseteq \tilde{\lat}$ where~$\tilde{\lat}$ has a second moment~$\Px/2$. The transmitted signal is the combination of the two lattice codewords
\begin{equation}
\tilde{\xv} = \big[\tilde{\tv}^{(R)} - \tilde{\sv}^{(R)} - \tilde{\dv}^{(R)} \big] \text{mod} \tilde{\lat} 
+ j \big[\tilde{\tv}^{(I)} - \tilde{\sv}^{(I)} - \tilde{\dv}^{(I)} \big] \text{mod} \tilde{\lat},
\label{sig_tx_complex}
\end{equation}
with $\Ex \big[ ||\tilde{\xv}||^2 \big] \leq n\Px$, and the dithers $\tilde{\dv}^{(R)},\tilde{\dv}^{(I)}$ are independent. 

{\em Decoding:} The equalization matrix at time~$i$ is given by
\begin{equation}
\tilde{\Um}_i= (\frac{\Px}{M}+\Ps) \big((\frac{\Px}{M}+\Ps) \tilde{\Hm}_i \tilde{\Hm}_i^H + \Pw \Ix_{M n} \big)^{-1} \tilde{\Hm}_i.
\label{eq:ui_MSE_complex}
\end{equation}

Following  MMSE equalization and dither removal similar to~\eqref{sig_rx_2}, the real and imaginary equivalent channels at the receiver are as follows
\begin{equation}
\tilde{\yv}'^{(R)} = \tilde{\tv}^{(R)} + \boldsymbol{\lambda'}  + \tilde{\zv}^{(R)}  ~~ , ~~
\tilde{\yv}'^{(I)} = \tilde{\tv}^{(I)} + \boldsymbol{\lambda''} + \tilde{\zv}^{(I)} ,
\label{sig_rx_complex}
\end{equation}
where $\boldsymbol{\lambda'}$,$\boldsymbol{\lambda''} \in \tilde{\lat}$. $\tilde{\zv}^{(R)},\tilde{\zv}^{(I)}$ are the equivalent noise components over the real and imaginary channels, respectively. 
Hence, the complex-valued channel is transformed to two parallel real-valued channels, where the decoder recovers the lattice points~$\tilde{\tv}^{(R)}$ and~$\tilde{\tv}^{(I)}$ independently over the real and imaginary domains using the following decision regions
\begin{align}
\genvoronoi^{(R)} = \genvoronoi^{(I)} \triangleq \bigg\{ & \vv \in \real^{M n} \, : \, ||\vv||^2 \leq (1+\epsilon) \frac{n}{2} \,
\twocolbreak
\text{tr} \Big( \Ex \Big [ \big( \frac{1}{\frac{\Px}{M}+\Ps} \Ix_{M} + \frac{1}{\Pw} \tilde{\Hm}^H \tilde{\Hm} \big)^{-1} \Big ]  \Big)\bigg\} \, .
\label{voronoi_MIMO_complex}
\end{align} 
The rates achieved on each channel are then given by
\begin{align}
& \tilde{R}^{(R)} = \tilde{R}^{(I)} < 
\twocolbreak
 \bigg ( \frac{-1}{2} \log \det \Big( \Ex \big[ \big( \frac{\Px}{\Px+ M \Ps} \Ix_{M} + \frac{\Px}{M \Pw} \tilde{\Hm}^H \tilde{\Hm}  \big)^{-1}  \big] \Big) \bigg ) ^+, 
\label{R_lat_complex_per}
\end{align}
 and $\tilde{R} \triangleq \tilde{R}^{(R)} + \tilde{R}^{(I)}$ is the achievable rate for the complex-valued channel given in~\eqref{R_lat}.
\end{proof}

\begin{remark}
It was shown  that for any rate satisfying~\eqref{R_lat}, the received point falls with high probability within the decision region~\eqref{voronoi_MIMO} of the transmitted lattice point. The fact that the decision region is spherical implies that the true lattice point has the shortest Euclidean distance to the received point. This indicates that the Euclidean lattice decoder 
\begin{equation}
\hat{\tv}= \big [ \text{arg} \min_{\tv' \in \lat_1} ||\yv'- \tv'||^2 \big ] \text{ mod} \lat 
\label{ML_eqn_simple}
\end{equation} 
yields similar performance to that of the ambiguity decoder. However, unlike the lattice decoders in~\cite{DMT_lattices,Lattice_MIMO_DPC_C}, the decoding rule in~\eqref{ML_eqn_simple} is independent of the channel realizations and hence a universal codebook and decoder are guaranteed for all realizations of a given fading distribution.
\end{remark}

We compute bounds on the gap between the outer and inner bounds of the capacity in~\eqref{capacity_outer} and~\eqref{R_lat}. 
For convenience we define the SNR per transmit antenna $\SNR \triangleq \frac{\Px}{M \Pw}$.

\begin{corollary}
For the fading dirty paper channel in~\eqref{sig_Rx}, the gap between the lattice coding rate~\eqref{R_lat} and the capacity outer bound~\eqref{capacity_outer} for any $\Ps>0$ is upper bounded by
\begin{itemize}
\item $N \geq M$ and $\SNR \geq 1$: For any channel for which all elements of $\Ex \big[ ( \Hm^H \Hm  )^{-1} \big] < \infty$ 
\begin{equation}
\gap < \log  \det \Big( \big( \Ix_M + \Ex [ \Hm^H \Hm ] ) \, \Ex \big[  ( \Hm^H \Hm  )^{-1}  \big] \Big). 
\label{gap_NtNr_eq}
\end{equation}
\item $N > M$ and $\SNR \geq 1$: When the elements of~$\Hm$ are i.i.d. complex Gaussian with zero mean and unit variance, 
\begin{equation}
\gap < M \, \log \big( 1+\frac{M+1}{N-M} \big).
\label{gap_NtNr_uneq}
\end{equation}
\item $N=M=1$: Under Nakagami-$m$ fading with~$m>1$ and $\Ex [|h|^2]=1$,
\begin{equation}
\gap < 1 + \log \big ( 1 + \frac{1}{m-1} \big ).
\label{gap_SNRnak}
\end{equation}
\item $N=M=1$: Under Rayleigh fading with $\Ex [|h|^2]=1$, 
\begin{equation}
\gap < 1.48 + \log \big ( \log (1+ \kappa) \big ),
\label{gap_SNRgaus}
\end{equation}
where~$\kappa \triangleq \max \{ \frac{\Px}{\Pw} , \frac{\Ps}{\Pw} , 1 \}$.
\end{itemize}
\label{cor:gap}
\end{corollary}

\begin{proof}
See Appendix~\ref{appendix:gap}.
\end{proof}

Under Rayleigh fading with $M=N=1$ the gap expression in~\eqref{gap_SNRgaus} varies with~$\Px,\Ps$. However, it can be shown that $\lim_{\Px \to \infty} \frac{\gap}{C} \to 0$ for any fixed ratio~$\frac{\Px}{\Ps} \, $. Nevertheless, when $M<N$, $\gap$ is independent of~$\Px,\Ps \, $, and also vanishes when $N \gg M$ even at finite~$\Px$. This result implies that lattice coding and decoding- along with a channel independent decision rule- approach the capacity of the Rayleigh-fading MIMO channel with $N>M$.

 In Fig.~\ref{fig:gap_MIMO}, the bound on the gap to capacity is plotted for different antenna configurations, which holds for all $\SNR \geq 1$. The gap vanishes when~$N \gg M$. For the square MIMO channel with~$M=N=2$, lattice coding rates are compared with the DPC rates in~\eqref{capacity_inner_MIMO} as well as the capacity outer bound in~\eqref{capacity_outer}, and the gap to capacity of the lattice scheme are plotted in Fig.~\ref{fig:rates_2x2}. Simulation results for the single-antenna case are also provided under Nakagami-$m$ fading with $m=2$ in Fig.~\ref{fig:gap_Nakagami} and under Rayleigh fading in Fig.~\ref{fig:gap_Rayleigh}.

\begin{figure}
\centering
\includegraphics[width=\Figwidth]{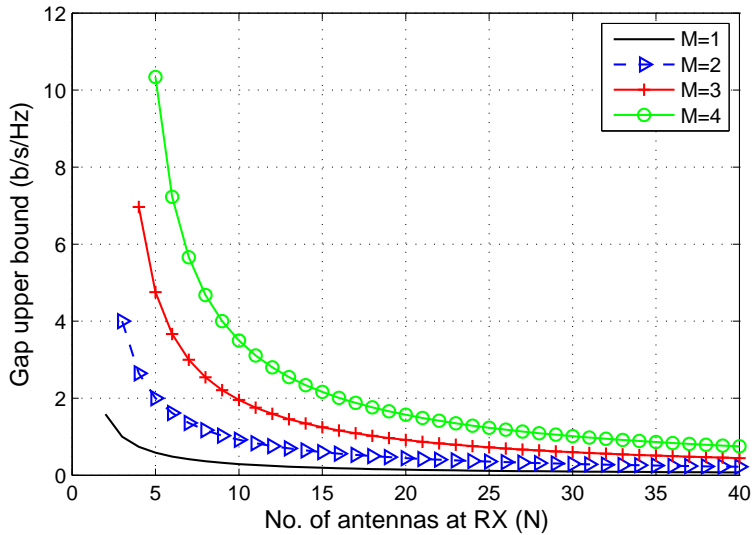}
\caption {The gap expression in~\eqref{gap_NtNr_uneq} vs. $N$, at $M=1,2,3$ and~$4$, for all~$\SNR \geq 1$.} 
\label{fig:gap_MIMO} 
\end{figure}

\begin{figure*}
\centering
\includegraphics[width=1\textwidth]{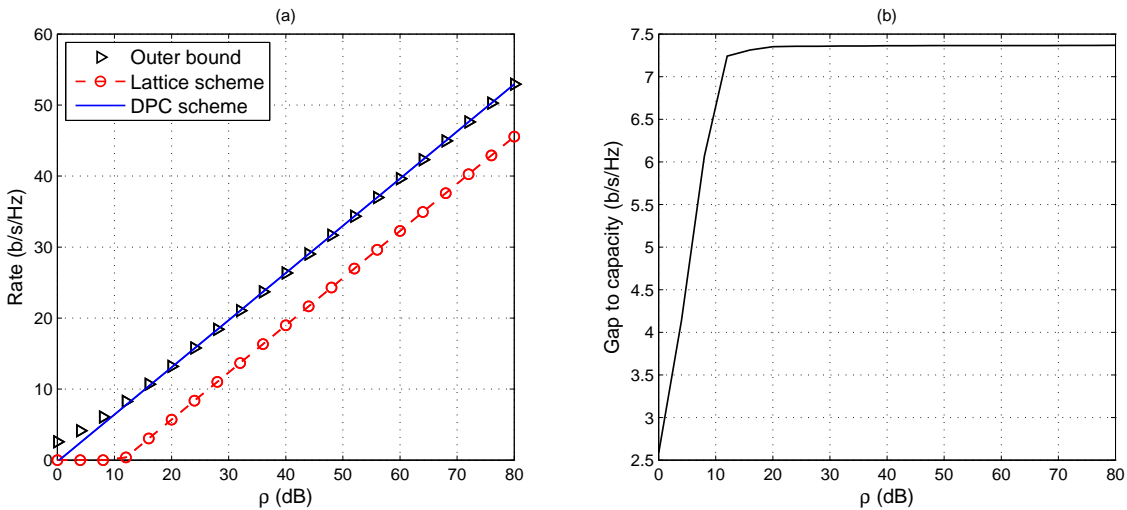}
\caption {\footnotesize{(a) Rates achieved by lattice codes vs. DPC vs. outer bound under Rayleigh fading with $M=N=2$ and $\Ps =80 \, \text{dB}$. (b)~Gap to capacity of the lattice scheme.} 
\label{fig:rates_2x2} } 
\end{figure*}

\begin{figure}
\centering
\includegraphics[width=\Figwidth]{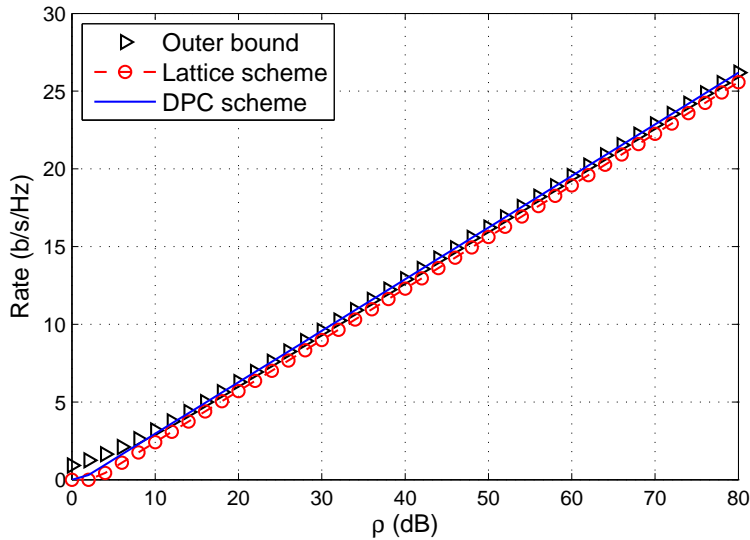}
\caption {Rates achieved using lattice codes vs. DPC vs. outer bound under Nakagami-$m$ fading with $m=2$ and $\Ps =80 \, \text{dB}$.} 
\label{fig:gap_Nakagami} 
\end{figure}

\begin{figure}
\centering
\includegraphics[width=\Figwidth]{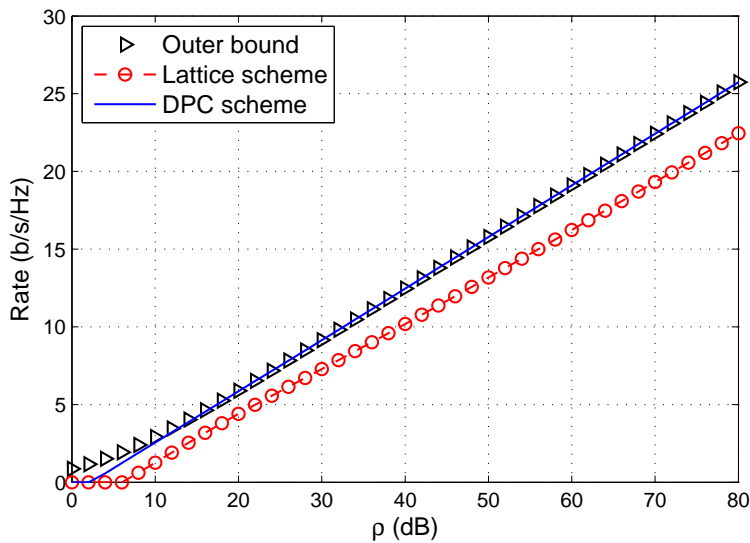}
\caption {Rates achieved using lattice codes vs. DPC vs. outer bound under Rayleigh fading and $\Ps =80 \, \text{dB}$.} 
\label{fig:gap_Rayleigh}  
\end{figure}


\section{Fading Broadcast Channel} 
\label{sec:BC}

We first consider a two-user broadcast channel where the channel coefficients of Receiver~1 are quasi-static, and that of Receiver~2 are stationary and ergodic. The transmitter and the two receivers have $M, N_1, N_2$ antennas, respectively.  The received signals are given by
\begin{align}
\yv_{1,i} = & \Gm \xv_{1,i} + \Gm \xv_{2,i} +\wv_{1,i}    \nonumber \\
\yv_{2,i} = & \Hm_i \xv_{1,i} + \Hm_i \xv_{2,i} +\wv_{2,i} \, .
\label{sig_Rx_BC}
\end{align}
Each receiver has its own CSIR, but not global CSIR. The transmitter
power constraint for the two signals is $\Ex [||\xv_1||^2] \leq n
\alpha \Px $ and $\Ex [||\xv_2||^2] \leq n (1-\alpha) \Px $, with
$\alpha \in [0,1]$ and $n$ represents the time duration of each codeword. The
noise terms $\wv_1, \wv_2$ are zero mean i.i.d. circularly-symmetric
complex Gaussian with variances~$P_{w_1}, P_{w_2}$, respectively.
A set of achievable rates for this channel under lattice coding and decoding are as follows
\begin{theorem}
For the two-user broadcast channel given in~\eqref{sig_Rx_BC}, lattice coding and decoding achieve 
\begin{align}
R_1 < &  \,  \log \det \Big( \Ix_M  +  \frac{\alpha \Px}{M}  \Gm^H  \boldsymbol{\Phi}^{-1} \Gm \Big)    
\label{R_lat_BC_1}   \\
R_2 < &  \Big ( -  \log \det \Big( \frac{1}{1-\alpha} \Ex \big[ \big( \Ix_M + \frac{\Px}{M P_{w_2}} \Hm^H \Hm  \big)^{-1}  \big] \Big) \Big ) ^+, 
\label{R_lat_BC_2}
\end{align} 
where $\boldsymbol{\Phi} \triangleq  \frac{(1-\alpha) \Px}{M} \Gm \Gm^H + P_{w_1} \Ix_{N_1}$.
\label{theorem:BC_lat}
\end{theorem}

\begin{proof}
\textit{Receiver~$1$:} The transmitter emits a superposition of two codewords, i.e., $\xv = \xv_1 + \xv_2$, where Receiver~$1$ decodes $\xv_1$ while treating~$\xv_2$ as noise. Hence, with respect to Receiver~$1$, the channel is a special case of the dirty paper channel with~$\Ps=0$,  and colored noise given by $\Gm \xv_{2,i} +\wv_{1,i} $. The equalization matrix is then time invariant, given by
\begin{equation}
\Um_1 = \frac{\alpha \Px}{M} \big( \frac{\Px}{M} \Gm \Gm^H + \Ix_{N_1} \big)^{-1} \Gm \, . 
\label{U_MSE_BC}
\end{equation}
Since the channel is fixed, we use an ellipsoidal decision region, given by 
\begin{equation}
\genvoronoi \triangleq \Big\{ \vv \in \real^{M n} \, : \, \vv^T \boldsymbol{\cov}_1^{-1} \vv \leq  (1+\epsilon)  n   \Big\}, 
\label{voronoi_BC}
\end{equation} 
where~$\boldsymbol{\cov}_1$ is an $Mn \times Mn$ block-diagonal matrix whose diagonal blocks are equal, and given by
\begin{equation}
\boldsymbol{\cov}_{1,i} \triangleq  \frac{\alpha \Px}{M} \big( \Ix_{M} + \frac{\alpha \Px}{M} \Gm^H \boldsymbol{\Phi}^{-1} \Gm  \big)^{-1}.
\label{cov_BC}
\end{equation} 
Following in the footsteps of the proof of Theorem~\ref{theorem:ptp_MIMO}, it can be shown that~$R_1$ satisfies~\eqref{R_lat_BC_1}.

\textit{Receiver~$2$:} Since~$\xv_1$ is known non-causally at the transmitter, the channel between the transmitter and Receiver~$2$ is equivalent to an ergodic fading dirty paper channel with $\Ps= \alpha \Px$, where  $R_2$ is given by~\eqref{R_lat_BC_2} from Theorem~\ref{theorem:ptp_lat}. The remainder of the rate region is obtained by varying~$\alpha$. 
\end{proof}

In the absence of CSIT, the capacity of the fading MIMO BC remains unknown. In Fig.~\ref{fig:rate_MIMO_BC_Rayleigh} we compare the rate region of the lattice coding scheme to the time-sharing inner bound as well as a version of Costa's DPC under non-causal CSIT and white-input covariance. In this non-causal scheme, Receiver~$1$ decodes its message while treating $\xv_2$ as noise. Since $\{ \Hm_i \xv_{1,i} \}_{i=1}^n$ are known non-causally at the transmitter, DPC totally removes the interference at Receiver~$2$. The rate region is then given by\footnote{It is unknown whether the rate region in~\eqref{C_lat_BC} is an outer bound for the capacity region of the channel in~\eqref{sig_Rx_BC}.}
\begin{align}
\check{R}_1 < &  \,  \log \det \Big( \Ix_M  +  \frac{\alpha \Px}{M}  \Gm^H  \boldsymbol{\Phi}^{-1} \Gm \Big)    
\nonumber   \\
\check{R}_2 < &  \,  \Ex \Big[  \log \det \big( \Ix_{M} + \frac{(1-\alpha) \Px}{M P_{w_2}} \Hm^H \Hm  \big)  \Big].
\label{C_lat_BC}
\end{align} 
In Fig.~\ref{fig:rate_MIMO_BC_Rayleigh} we compute the rates through Monte-Carlo simulations when $M=N_1=2, N_2=4$, and the channel of Receiver~$2$ is Rayleigh faded.\footnote{Jafar and Goldsmith~\cite{Goldsmith_BC_CSIR} showed that increasing~$M$ does not increase the capacity of the broadcast channel with isotropic fading and CSIR. Hence, we focus in our simulations on cases where $M \leq \min \{ N_1,N_2 \}$.} 
For the special case of single-antenna nodes, the rates of the lattice coding scheme are plotted in Fig.~\ref{fig:rate_SISO_BC_Nakagami} and compared with the time sharing inner bound as well as the Tuninetti-Shamai rate region for the two-user fading BC~\cite{Shamai_BC_C}.\footnote{The authors of~\cite{Shamai_BC_C} conjecture that their inner bound is tight.}
The results are also compared with the white-input BC capacity  with CSIT~\cite{BC_ergodic}. For the single-antenna case we assume the channel of Receiver~1 has unit gain, i.e.,~$|g|=1$. Note that unlike both~\cite{Shamai_BC_C,BC_ergodic}, the proposed lattice scheme presumes each receiver is oblivious to the codebook designed for the other receiver.

\begin{figure}
\centering
\includegraphics[width=\Figwidth]{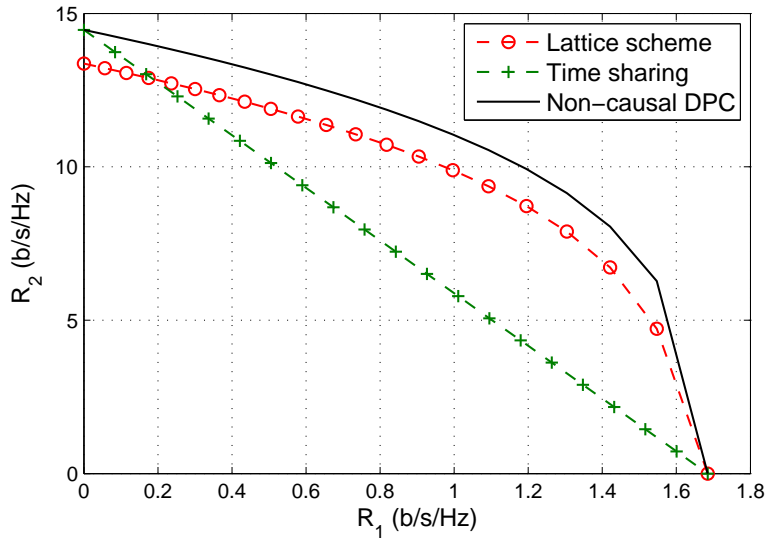}
\caption {Rate regions for the BC with $M=N_1=2$ and $N_2=4$, where~$\frac{\Px}{P_{w_1}}= 0 \, \text{dB}$ and~$\frac{\Px}{P_{w_2}}= 20 \, \text{dB}$. Fading of second user is Rayleigh.} 
\label{fig:rate_MIMO_BC_Rayleigh} 
\end{figure}

\begin{figure}
\centering
\includegraphics[width=\Figwidth]{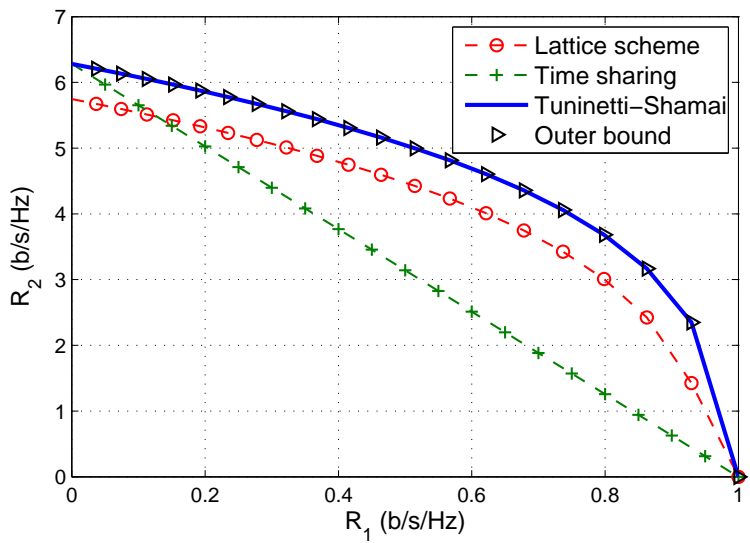}
\caption {Rate regions for the single-antenna BC, where~$\frac{\Px}{P_{w_1}}= 0 \, \text{dB}$ and~$\frac{\Px}{P_{w_2}}= 20 \, \text{dB}$. Fading of second user is Nakagami-$m$ with $m=2$.} 
\label{fig:rate_SISO_BC_Nakagami} 
\end{figure}

In addition, we study the two-user broadcast channel with CSIR, where the fading processes of the two users are stationary, ergodic and independent of each other, as follows
\begin{align}
\yv_{1,i} = & \Hm_{1,i} \xv_{1,i} + \Hm_{1,i} \xv_{2,i} +\wv_{1,i}    \nonumber \\
\yv_{2,i} = & \Hm_{2,i} \xv_{1,i} + \Hm_{2,i} \xv_{2,i} +\wv_{2,i}.
\label{sig_Rx_BC_F}
\end{align}

\begin{theorem}
For the two-user broadcast channel given in~\eqref{sig_Rx_BC_F}, lattice coding and decoding achieve
\begin{align}
R_1 < &  \Big ( -  \log \det \Big( \Ex \big[ \big( \Ix_M + \frac{ \alpha \Px}{M P_{w_1}} \Hm_1^H  \boldsymbol{\Phi}^{-1} \Hm_1  \big)^{-1}  \big] \Big) \Big ) ^+,   
\label{R_lat_BC_F_1}   \\
R_2 < &  \Big ( -  \log \det \Big( \frac{1}{1-\alpha} \Ex \big[ \big( \Ix_M + \frac{\Px}{M P_{w_2}} \Hm_2^H \Hm_2  \big)^{-1}  \big] \Big) \Big ) ^+, 
\label{R_lat_BC_F_2}
\end{align} 
where $\boldsymbol{\Phi} \triangleq  \frac{(1-\alpha) \Px}{M} \Hm_1 \Hm_1^H + P_{w_1} \Ix_{N_1}$.
\label{theorem:BC_lat_F}
\end{theorem}

\begin{proof}
The achievability proof of the rate of Receiver~2 in~\eqref{R_lat_BC_F_2} is identical to that of Theorem~\ref{theorem:BC_lat}. At Receiver~1, the received signal is multiplied by a time-varying equalization matrix, given by
\begin{equation}
\Um_{1,i} = \frac{\alpha \Px}{M} \big( \frac{\Px}{M} \Hm_{1,i} \Hm_{1,i}^H + \Ix_{N_1} \big)^{-1} \Hm_{1,i} \, ,  
\label{U_MSE_BC_F}
\end{equation}
with spherical decision region as follows
\begin{equation}
\genvoronoi \triangleq  \Big\{ \vv \in \real^{M n} \, : \, ||\vv||^2 \leq  (1+\epsilon)  n  \frac{\alpha \Px}{M} 
 \text{tr} \Big( \Ex \Big [  \big( \Ix_M + \frac{\alpha \Px}{ M} \Hm_1^H  \boldsymbol{\Phi}^{-1}   \Hm_1  \big)^{-1}  \Big] \Big)   \Big\}. 
\label{voronoi_BC_F}
\end{equation} 
The remainder of the analysis resembles that in the proof of Theorem~\ref{theorem:BC_lat}, where it can be shown that~\eqref{R_lat_BC_F_1} is achievable.
\end{proof}

The rate region of the lattice scheme is plotted in Fig.~\ref{fig:rate_MIMO_BC_Nakagami_F} under Nakagami fading with $M=1$ and $N_1=N_2=2$, and compared with time sharing and DPC with non-causal CSIT.

\begin{figure}
\centering
\includegraphics[width=\Figwidth]{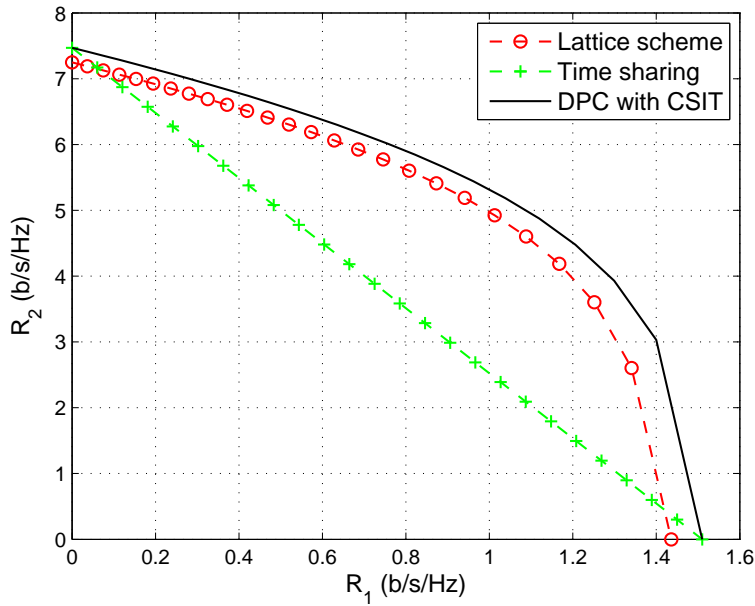}
\caption {Rate regions for the Nakagami-$m$ faded BC with $M=1$ and $N_1=N_2=2$, where~$\frac{\Px}{P_{w_1}}= 0 \, \text{dB}$ and~$\frac{\Px}{P_{w_2}}= 20 \, \text{dB}$. $m=2$.} 
\label{fig:rate_MIMO_BC_Nakagami_F} 
\end{figure}


\section{Conclusion} 
\label{sec:conc}

This paper studies the MIMO dirty paper channel in which the channel input and dirt experience stationary and ergodic fading with CSIR. It is shown that a variant of Costa's dirty paper coding achieves rates within $M$~bits of the capacity. Moreover, a lattice coding and decoding scheme is proposed that achieve rates within a constant gap to capacity for a wide range of fading distributions. More specifically, the gap to capacity diminishes as the number of receive antennas increases, even at finite SNR. The decision regions do not depend on the channel realizations, leading to simplifications.  The results imply that lattice coding and decoding  approach optimality for the fading dirty paper channel, and that the capacity of the fading dirty paper channel with CSIR approaches that of the point-to-point channel under some antenna configurations.  Moreover, the results are applied to MIMO broadcast channels under different fading scenarios and compared to capacity outer bounds. Simulations show that the proposed lattice coding scheme has near-capacity performance.


\appendices

\section{Proof of Lemma~\ref{lemma:z_dist}}
\label{appendix:z_dist}

We rewrite the noise expression~$\zv$ in~\eqref{eq:zi} in the form $\zv = \Am_d \xv + \Am_d \sv + \sqrt{\frac{\mu}{\Pw}} \Bm_d \wv$, where $\mu \triangleq \frac{\Px}{M}+\Ps$, and both $\Am_d$, $\Bm_d$ are block-diagonal matrices with diagonal blocks~$\Am_i,\Bm_i$, as follows
\begin{align}
\Am_i &\triangleq - \big ( \Ix_M + \frac{\mu}{\Pw} \Hm_i^T  \Hm_i  \big)^{-1},
\label{A_matrix} \\
\Bm_i & \triangleq \sqrt{\Pw} \mu \Hm_i^T ( \mu \Hm_i \Hm_i^T + \Pw \Ix_{N} )^{-1}.
\label{B_matrix}
\end{align}

Note that
\begin{equation}
\Am_d \Am_d^T + \Bm_d \Bm_d^T = ( \Ix_M + \frac{\mu}{\Pw} \Hm_d^T \Hm_d )^{-1}. 
\label{cov_sum}
\end{equation}
Since $\Hm_i$ is a stationary and ergodic process, $\Am_i$ and~$\Bm_i$ are stationary and ergodic processes as well. In the proceeding we omit the time index $i$ whenever it is clear from the context. Denote the ordered eigenvalues of the random matrix $\Hm^T \Hm$  by $\sigma_{H,1}^2, \ldots, \sigma_{H,M}^2$ (non-decreasing). Then the eigenvalue decomposition of~$\Hm^T \Hm$ is $\Hm^T \Hm \triangleq \Vm \Dm \Vm^T$, where $\Vm$ is a unitary matrix and $\Dm$ is a diagonal matrix whose unordered entries are $\sigma_{H,1}^2, \ldots, \sigma_{H,M}^2$. Owing to the isotropy of the distribution of~$\Hm$, $\Am \Am^T = \Vm (\Ix_M + \frac{\mu}{\Pw}  \Dm)^{-2} \Vm^T$ is unitarily invariant, i.e., $ \prob (\Am \Am^T) =  \prob (\check{\Vm} \Am \Am^T \check{\Vm}^T)$ for any unitary matrix~$\check{\Vm}$ independent of~$\Am$. As a result $\Vm$ is independent of~$\Dm$~\cite{random_matrix}. Hence,
\begin{align}
\Ex \big[ \Am \Am^T \big] & = \,  \Ex \big[ ( \Ix_M + \frac{\mu}{\Pw} \Hm^T \Hm  )^{-2}  \big]  \,
\twocolbreak
= \, \Ex \big[ \Vm (\Ix_M + \frac{\mu}{\Pw}  \Dm)^{-2} \Vm^T  \big]  \nonumber \\
& =  \Ex_{\Vm|\Dm} \big[ \Vm  \Ex_{\Dm} [ (\Ix_M + \frac{\mu}{\Pw}  \Dm )^{-2} ] \Vm^T  \big]  \,
\twocolbreak
= \, \Ex_{\Vm|\Dm} \big[ \Vm  \sigma_A^2 \Ix_M  \Vm^T  \big]  \,
= \,  \sigma_A^2 \Ix_M,
\label{A_cov_2}
\end{align}
where $\sigma_A^2 \triangleq \Ex_j \big[ \Ex_{\sigma_{\Hm,j}} [ \frac {1}{(1+ \frac{\mu}{\Pw} \sigma_{\Hm,j}^2)^2} ]   \big]$. Similarly, it can be shown that $\Ex \big[ \Bm \Bm^T \big] =  \sigma_B^2 \Ix_M$, where 
\begin{equation*}
\sigma_B^2 \triangleq \Ex_j \Big[ \Ex_{\sigma_{\Hm,j}} \big[ \frac {\frac{\mu}{\Pw} \sigma_{\Hm,j}^2}{(1+ \frac{\mu}{\Pw} \sigma_{\Hm,j}^2)^2} \big]   \Big].
\end{equation*}
 For convenience define $\sigma_z^2 \triangleq \sigma_A^2+\sigma_B^2$. Next, we  compute the autocorrelation of $\zv$ as follows
\begin{equation}
\boldsymbol{\cov_z} \triangleq  \Ex \big[  \zv \zv^T  \big] = \Ex \big[ \Am_d (\boldsymbol{\cov_x} + \Ps \Ix_{Mn}) \Am_d^T \big] + \mu \Ex \big[ \Bm_d  \Bm_d^T \big],
\label{cov_1}
\end{equation}	
where $\boldsymbol{\cov_x} \triangleq \Ex \big[ \xv \xv^T \big]$. Unfortunately, $\boldsymbol{\cov_x}$ is not known for all~$n$, yet it approaches $\frac{\Px}{M} \Ix_{M n}$ for large $n$, according to Lemma~\ref{lemma:x_distribution}. Hence one can rewrite 
\begin{align}
\boldsymbol{\cov_z} = & \, \underbrace { (\sigma_x^2 + \Ps) \Ex \big[ \Am_d \Am_d^T  +  \Bm_d  \Bm_d^T \big] }_{(\sigma_x^2 + \Ps) \, \sigma_z^2 \Ix_{M n}}  +
\twocolbreak
  \underbrace { \Ex \big[ \Am_d (\boldsymbol{\cov_x} - \sigma_x^2 \Ix_{M n}) \Am_d^T \big] +  \big( \frac{\Px}{M} - \sigma_x^2 \big) \Ex \big[ \Bm_d  \Bm_d^T \big] }_{\succ {\mathbf 0}}  \,  ,
\label{cov_2}
\end{align}	
where 
$\sigma_x^2 \triangleq \lambda_{\text{min}}(\boldsymbol{\cov_x}) - \delta$. 
It follows that $ \boldsymbol{\cov_z} \succ (\sigma_x^2 + \Ps) \, \sigma_z^2 \Ix_{M n} $, therefore
\begin{equation}
\boldsymbol{\cov_z}^{-1} \prec \frac{1}{(\sigma_x^2 + \Ps) \sigma_z^2} \Ix_{M n}.
\label{cov_3}
\end{equation}


To make noise calculations more tractable, we introduce a related noise variable that modifies the second term of $\zv$ as follows
\begin{equation}
\zv^* = \Am_d \xv + \Am_d \sv + \sqrt{\frac{\mu}{\Pw}} \Bm_d \wv + \sqrt{\frac{\frac{1}{M n} R_c^2 +\Ps}{\Pw} - \frac{\mu}{\Pw}} \, \Bm_d \wv^* ,
\label{z_MIMO_2}
\end{equation}
where $\wv^*$ is i.i.d. Gaussian with zero mean and unit variance, and $R_c$ is the covering radius of~$\voronoi$, and hence $\frac{1}{n} R_c^2 > \Px$. We now wish to bound the probability that $\zv^*$ is outside a sphere of radius $\sqrt{(1+\epsilon) M n ( \sigma_x^2 + \Ps) \sigma_z^2}$ for some~$\epsilon$ that vanishes with~$n$. First, we rewrite 
\begin{align}
||\zv^*||^2 = & \, \frac{\frac{1}{M n} R_c^2 +\Ps}{\Pw} \, \wv^T \Bm_d^T \Bm_d \wv  + \sv^T \Am_d^T \Am_d \sv 
\twocolbreak
+ 2 \sqrt{\frac{\frac{1}{M n} R_c^2 +\Ps}{\Pw}} \, \sv^T \Am_d^T \Bm_d \wv +  \xv^T \Am_d^T \Am_d \xv     \nonumber \\
 &  +  2 \sqrt{\frac{\frac{1}{M n} R_c^2 +\Ps}{\Pw}} \, \xv^T \Am_d^T \Bm_d \wv + 2 \xv^T \Am_d^T \Am_d \sv.
\label{z_norm}
\end{align}
We now bound the probability of deviation each of the terms on the
right hand side of~\eqref{z_norm} from its mean using the law of large
numbers. 
To begin with, the first term in~\eqref{z_norm} is the sum of $n$ terms of an ergodic sequence, where $\Ex [ \wv^T \Bm_d^T \Bm_d \wv ] = \text{tr} \big( \Ex [\Bm_d \wv \wv^T \Bm_d^T ] \big) = Mn \Pw \sigma_B^2$. Hence for any $\epsilon , \gamma \in (0,1)$ there exists sufficiently large $n$ so that
\begin{align}
 \prob & \Big( \frac{\frac{1}{M n} R_c^2 +\Ps}{\Pw} \, \wv^T \Bm_d^T \Bm_d \wv > ( R_c^2 + M n \Ps) \sigma_B^2 + M n \epsilon \Big) 
\twocolbreak
<\gamma,
\label{z_norm_w} 
\end{align}
Similarly,
\begin{equation}
\prob \big( \sv^T \Am_d^T \Am_d \sv > M n \Ps \sigma_A^2 + M n \epsilon \big) < \gamma,
\label{z_norm_s}
\end{equation}
and
\begin{equation}
\prob \Big( 2 \sqrt{\frac{\frac{1}{M n} R_c^2 +\Ps}{\Pw}} \, \sv^T \Am_d^T \Bm_d \wv > M n \epsilon \Big) < \gamma.
\label{z_norm_sw}
\end{equation}
The next term in~\eqref{z_norm} involves $\Am_d^T \Am_d$, a block-diagonal matrix with $\Ex \big[ \Am_d^T \Am_d \big] = \sigma_A^2 \Ix_{M n}$. Considering $\boldsymbol{\cov_x} \to \SNR \Ix_{M n}$ as $n \to \infty$, it can be shown using~\cite[Theorem 1]{weighted_avg} that $\frac{1}{||\xv||^2} \xv^T \Am_d^T \Am_d \xv \to \sigma_A^2$. More precisely,
\begin{equation}
 \prob \big(  \xv^T \Am_d^T \Am_d \xv   >  \sigma_A^2 ||\xv||^2  + M n \epsilon \big) < \gamma. 
\label{z_norm_xx}
\end{equation}
Since~$||\xv||^2 < R_c^2$, \eqref{z_norm_xx} implies
\begin{equation}
 \prob \big( \xv^T \Am_d^T \Am_d \xv  >  \sigma_A^2 R_c^2 + M n \epsilon \big) < \gamma.
\label{z_norm_x}
\end{equation}
The penultimate term in~\eqref{z_norm} can be bounded as follows\footnote{This term can be expressed as the sum of zero-mean and uncorrelated random variables to which the law of large numbers apply~\cite{LLN}.}
\begin{equation}
\prob \Big( 2 \sqrt{\frac{\frac{1}{M n} R_c^2 +\Ps}{\Pw}} \, \xv^T \Am_d^T \Bm_d \wv > M n \epsilon \Big) < \gamma,
\label{z_norm_xw}
\end{equation}
where the elements of~$\xv$ are also bounded. Similarly for the final term in~\eqref{z_norm},
\begin{equation}
\prob \big( 2 \xv^T \Am_d^T \Am_d \sv > M n \epsilon \big) < \gamma,
\label{z_norm_xs}
\end{equation}
In \eqref{z_norm_w} through \eqref{z_norm_xs}, $\epsilon,\gamma$ can
be made arbitrarily small by increasing~$n$. Moreover, for any fixed
$\epsilon, \gamma$ there is a sufficiently large $n$ so that
simultaneously all the above bounds are satisfied, because their number is finite and for each one a sufficiently large $n$ exists.

We now produce a union bound on all the terms above
\begin{equation}
\prob \big( ||\zv^*||^2 > (1+\epsilon')  ( R_c^2 + M n \Ps) \sigma_z^2  \big) < 6\gamma,  
\label{z_norm_final_1}
\end{equation}
where $\epsilon' \triangleq \frac{6\epsilon}{( R_c^2 + M n \Ps) \sigma_z^2}$. For sufficiently large~$n$, we can find $\frac{1}{M n}R_c^2 \leq (1+\epsilon) \frac{\Px}{M}$ for covering-good lattices and $\frac{\Px}{M} \leq (1+\epsilon) \sigma_x^2$ according to Lemma~\ref{lemma:x_distribution}. Then, take $\epsilon'' \triangleq (1+\epsilon)^2-1$, and any $\epsilon''' \leq (1+\epsilon')(1+\epsilon'') -1$, we have
\begin{align}
& \prob \big( \zv^{*T} \boldsymbol{\cov_z}^{-1} \zv^* > (1+\epsilon''')  M n   \big) 
\nonumber \\
< &  \prob \big( ||\zv^*||^2 > (1+\epsilon''')  n (M \sigma_x^2 + M \Ps) \sigma_z^2  \big)  
\label{z_norm_final_2} \\
< &  \prob \big( ||\zv^*||^2 > (1+\epsilon'''')  n (\Px + M \Ps) \sigma_z^2  \big)  
\nonumber \\
= &  \prob \Big( ||\zv^*||^2   
> (1+\epsilon'''')  \, \text{tr} \big( \Ex \big [ \big( \frac{1}{\frac{\Px}{M}+\Ps} \Ix_{Mn} + \frac{1}{\Pw} \Hm_d^T \Hm_d \big)^{-1} \big ] \big)    \Big)   
\label{z_norm_final_3} \\
= &  \prob \Big( ||\zv^*||^2   
> (1+\epsilon'''') n \, \text{tr} \big( \Ex \big [ \big( \frac{1}{\frac{\Px}{M}+\Ps} \Ix_M + \frac{1}{\Pw} \Hm^T \Hm \big)^{-1} \big ] \big)    \Big)  \nonumber \\
< & ~  6\gamma,  \nonumber
\end{align}
where~\eqref{z_norm_final_2} holds from~\eqref{cov_3} and~\eqref{z_norm_final_3} holds since $\Ex \big[ (\Ix_{M n} + \SNR \Hm_d^T \Hm_d)^{-1}  \big] = \sigma_z^2 \Ix_{M n}$, according to~\eqref{cov_sum}. The final step is to show that $||\zv^*|| \to ||\zv||$ as $n \to \infty$, where $\zv^*-\zv= \sqrt{\frac{\frac{1}{M n} R_c^2 +\Ps}{\Pw} - \frac{\mu}{\Pw}} \, \Bm_d \wv^*$. From the structure of~$\Bm_d$, the norm of each of its rows is less than~$M$, and hence the variance of each of the elements of~$\Bm_d \wv^*$ is no more than~$M$. Since $\lim_{n \to \infty} \frac{1}{M n} R_c^2 = \frac{\Px}{M}$ for a covering-good lattice, it can be shown using Chebyshev's inequality~\cite{kobayashi} that the elements of $\sqrt{\frac{\frac{1}{M n} R_c^2 +\Ps}{\Pw} - \frac{\mu}{\Pw}} \, \Bm_d \wv^*$ vanish and $|z_i^*-z_i| \to 0$ as $n \to \infty$ for all $i \in \{ 1, \ldots, n \}$, as follows
\begin{align}
\prob \big(|z_i^*-z_i| \geq \gamma^* \, \kappa \big) & \leq \frac{1}{\kappa^2}, ~~~~~~ for~all~ \kappa > 0,
\nonumber \\
\prob \big(|z_i^*-z_i| \geq \sqrt{\gamma^*} \,  \big) & \leq \gamma^* \,, 
\label{chebyshev}
\end{align}
where $\gamma^* \triangleq \sqrt{M \, \big( \frac{\frac{1}{M n} R_c^2 +\Ps}{\Pw} - \frac{\mu}{\Pw} \big)}$ vanishes with~$n$ and~\eqref{chebyshev} follows when $\kappa =\frac{1}{\sqrt{\gamma^*}}$.  This concludes the proof of Lemma~\ref{lemma:z_dist}.


\section{Proof of Corollary~\ref{cor:gap}}
\label{appendix:gap}

\subsection{$N \geq M$ and $\Ex \big[ (\Hm^H \Hm)^{-1} \big] < \infty$}
\label{sec:gap_NtNr_eq}

\begin{align}
\gap = & \, C-R  \nonumber  \\
= &  \Ex \big[ \log \det (  \Ix_{M} + \SNR \Hm^H \Hm ) \big] 
\twocolbreak
-  \Big ( - \log \det \Big( \Ex \big[ ( \frac{\Px}{\Px+M \Ps} \Ix_{M} + \SNR \Hm^H \Hm)^{-1} \big] \Big) \Big ) ^+  \nonumber \\
\leq &  \Ex \big[ \log \det (  \Ix_{M} + \SNR \Hm^H \Hm ) \big] 
\twocolbreak
+  \log \det \Big( \Ex \big[ ( \frac{\Px}{\Px+M \Ps} \Ix_{M} + \SNR \Hm^H \Hm)^{-1} \big] \Big) \nonumber \\
\leq & \log \det \Big( \Ix_{M} + \SNR \Ex [ \Hm^H \Hm ] \Big)  
\twocolbreak
+  \log \det \Big( \Ex \big[ ( \frac{\Px}{\Px+M \Ps} \Ix_{M} + \SNR \Hm^H \Hm)^{-1} \big] \Big) 
\label{gap_A_Jen} \\
< &  \log \det \Big( \Ix_{M} + \SNR \Ex [ \Hm^H \Hm ] \Big) 
+  \log \det \Big( \Ex \big[ ( \SNR \Hm^H \Hm)^{-1} \big] \Big) 
\label{gap_A_det} \\
= & \log \det \Big(  \big( \frac{1}{\SNR} \Ix_{M} + \Ex [ \Hm^H \Hm ] \big)  \Ex \big [ (\Hm^H \Hm)^{-1} \big ] \Big)   \nonumber \\ 
\leq & \log \det \Big(  \big( \Ix_{M} + \Ex [ \Hm^H \Hm ] \big)  \Ex \big [ (\Hm^H \Hm)^{-1} \big ] \Big),
\label{gap_A_power}
\end{align}
where~\eqref{gap_A_Jen},\eqref{gap_A_det} follow since $\log \det (\Am)$ is a concave and non-decreasing function over the set of all positive definite matrices~\cite{convex}. \eqref{gap_A_power} follows since $\SNR \geq 1$.


\subsection{$N > M$ and $\Hm$ is Gaussian}
\label{sec:gap_NtNr_uneq}

\begin{lemma}~\cite[Section~V]{Wishart}
For an i.i.d. complex Gaussian~$N \times M$ matrix~$\Gm$ whose elements have zero mean, unit variance and $N > M$, then $\Ex \big[ (\Gm^H \Gm)^{-1} \big] = \frac{1}{N-M} \, \Ix_N$. \quad \QEDB
\label{lemma:wishart}
\end{lemma}
It follows
\begin{align}
\gap  & <  \, \log \det \Big(  \big( \Ix_{M} + \Ex [ \Hm^H \Hm ] \big)  \Ex \big [ (\Hm^H \Hm)^{-1} \big ] \Big)  
\label{gap_B_A}  \\
& =  \log \det \big(  (1+N) \, \frac{1}{N-M} \Ix_{M}  \big) 
\label{gap_B_Wishart}  \\
& = \, M \log \big ( 1 + \frac{M+1}{N - M} \big ) , \nonumber
\end{align}
where~\eqref{gap_B_A},\eqref{gap_B_Wishart} follow from~\eqref{gap_A_power} and Lemma~\ref{lemma:wishart}, respectively.


\subsection{$N=M=1$ and $|h|$ is  Nakagami-$m$ with $m>1$}
\label{sec:gap_SNRnak}

The Nakagami-$m$ distribution with~$m>1$ satisfies the condition~$\Ex \big[ \frac{1}{|h|^2} \big] < \infty $, and hence $\gap$ is a special case of~\eqref{gap_A_power}. When $\Ex [|h|^2 ]=1$ and $\SNR \geq 1$, then
\begin{align*}
\Ex \big[ \frac{1}{|h|^2} \big] = &   
\frac{2m^m}{\Gamma(m)} \,  \int_{0}^{\infty} {\frac{1}{x^2} \, x^{2m-1} e^{-mx^2} dx }   \\
 =&   \frac{2m^m}{\Gamma(m)} \, \frac{1}{2 m^{m-1}} \, \int_{0}^{\infty} { y^{m-2} e^{-y} dy }   \\
 =&   m \,  \frac{\Gamma(m-1)}{\Gamma(m)}   \, = \,  1 + \frac{1}{m-1} \, .
\end{align*}
 Hence, $\gap < \log \big(  ( 1 + \Ex [ |h|^2 ] )  \Ex \big [ \frac{1}{|h|^2} \big ] \big) = 1+ \log \big (  1+ \frac{1}{m-1} \big )$.

The case where $\SNR <1$ is trivial, since\,%
\footnote{The result in~\eqref{gap_low_SNR} holds for any fading distribution with~$\Ex \big[ |h|^2 \big]=1$.}
\begin{equation}
\gap < C = \Ex \big[\log (1 + \SNR |h|^2) \big] \leq \log \big(1 + \SNR \Ex[|h|^2] \big) < 1,
\label{gap_low_SNR}
\end{equation}
and hence~$\gap <1+ \log \big (  1+ \frac{1}{m-1} \big )$ is universal for all~$\SNR$.


\subsection{$N=M=1$ and $|h|$ is Rayleigh}
\label{sec:gap_SNRgaus}

\begin{lemma} \cite[Section~5.1]{Handbook_math}
For any~$z>0$, 
\begin{equation*}
\bar{E}_1(z) \triangleq \int_{z}^{\infty}{ \frac{e^{-t}}{t} dt} < \frac{ e^{-z} }{\log e}  \, \log(1+\frac{1}{z}).
\qquad \qquad \qquad \QEDB
\end{equation*}
\label{lemma:expint}
\end{lemma}

When $|h|$ is a Rayleigh random variable, $|h|^2$ is exponentially distributed with~$\Ex [|h|^2 ]=1$. For the case where~$\SNR \geq 1$, 
\begin{align}
\gap \leq & \log \Big ( \big(1 + \SNR \Ex[|h|^2] \big) \Ex \big[ \frac{1}{\SNR |h|^2 + \frac{\SNR}{\frac{\Px}{\Pw}+\frac{\Ps}{\Pw}}} \big] \Big )  
\label{gap_D_A} \\
= & \log \Big ( (\frac{1}{\SNR} + 1 ) \, \Ex \big[ \frac{1}{ |h|^2 + \frac{1}{\frac{\Px}{\Pw}+\frac{\Ps}{\Pw}}} \big] \Big )  \nonumber  \\
\leq & 1+ \log \Big (  \Ex \big[ \frac{1}{ |h|^2 + \frac{1}{\frac{\Px}{\Pw}+\frac{\Ps}{\Pw}}} \big] \Big )   
\label{gap_D_power} \\ 
 \leq &  \, 1+ \log \Big (  \Ex \big[ \frac{1}{ |h|^2 +  \frac{1}{2 \kappa} } \big] \Big )     \nonumber  \\
= & 1+ \log \big (\int_{0}^{\infty}{\frac{e^{-x}}{x+\frac{1}{2 \kappa}}  dx} \big)  \nonumber \\ 
= & 1+ \log \big( e^{\frac{1}{2 \kappa}} \, \int_{\frac{1}{2 \kappa}}^{\infty}{\frac{e^{-y}}{y}  dy} \big)   \nonumber  \\
< & 0.48 + \log \big ( \log ( 1 + 2 \kappa ) \big )  
\label{gap_D_expint}  \\
  < & \, 1.48 + \log \big ( \log ( 1 + \kappa ) \big ) \, ,  
\label{gap_D_final}	
\end{align}
where~\eqref{gap_D_A} follows from~\eqref{gap_A_Jen}, \eqref{gap_D_power} follows since $\SNR \geq 1$ and \eqref{gap_D_expint} follows from Lemma~\ref{lemma:expint}. Recall $\kappa \triangleq \max \{ \frac{\Px}{\Pw} , \frac{\Ps}{\Pw} , 1 \}$. When~$\SNR <1$, the gap to capacity is within one bit, and hence \eqref{gap_D_final} is also an upper bound for the gap in this regime.


\bibliographystyle{IEEEtran}
\bibliography{IEEEabrv,References-v3}

\end{document}